\documentclass[10pt,letter]{article}

\usepackage{amsmath}
\usepackage{amsthm}
\usepackage{amssymb}
\usepackage{booktabs}
\usepackage{multirow}

\usepackage{enumitem}

\setlist[itemize]{leftmargin=*,itemsep=0pt}

\usepackage{fullpage}

\newcommand{\shortversion}[1]{}

\newcommand{\proofsketchversion}[1]{}
\newcommand{\fullproofversion}[1]{#1}

\usepackage{tikz}
\usetikzlibrary{shapes,shapes.callouts,shadows,arrows,backgrounds,%
matrix,patterns,arrows,decorations.pathmorphing,decorations.pathreplacing,%
positioning,fit,calc,decorations.text%
}

\newenvironment{myquote}{\begin{quote}\setlength{\parskip}{0pt}}{\end{quote}}

\newtheorem{LEM}{Lemma}
\newtheorem{THE}{Theorem}

\newtheorem{CLM}{Claim}

\newtheorem{PRO}{Proposition}

\def\hy{\hbox{-}\nobreak\hskip0pt}

\newcommand{\SB}{\{\,}%
\newcommand{\SM}{\;{|}\;}%
\newcommand{\SE}{\,\}}%

\let\phi=\varphi
\let\epsilon=\varepsilon

\newcommand{\CCC}{\mathcal{C}}

\newcommand{\SSS}{\mathcal{S}}

\newcommand{\NP}{\text{\normalfont NP}}
\newcommand{\paraNP}{\text{\normalfont para-NP}}

\newcommand{\coNP}{\text{\normalfont coNP}}

\newcommand{\paracoNP}{\text{\normalfont para-coNP}}
\newcommand{\FPT}{\text{\normalfont FPT}}

\newcommand{\W}[1][xxxx]{\text{\normalfont W[#1]}}

\newcommand{\STB}{\text{\normalfont stb}}
\newcommand{\ADM}{\text{\normalfont adm}}
\newcommand{\PRF}{\text{\normalfont prf}}
\newcommand{\COM}{\text{\normalfont com}}

\newcommand{\RBG}{\text{\normalfont rbg}}
\newcommand{\STG}{\text{\normalfont stg}}
\newcommand{\SEM}{\text{\normalfont sem}}
\newcommand{\bsdist}{\text{\normalfont dist}}

\newcommand{\MCC}{\textsc{Multicolored Clique}}
\newcommand{\HS}{\textsc{Hitting Set}}
\newcommand{\Small}{\textsc{$\sigma$-Small}}
\newcommand{\Repair}{\textsc{$\sigma$-Repair}}
\newcommand{\Adjust}{\textsc{$\sigma$-Adjust}}
\newcommand{\Center}{\textsc{$\sigma$-Center}}

\newcommand{\FOIS}{\textup{CF}}
\newcommand{\FOADM}{\ADM}
\newcommand{\FOCOM}{\COM}
\newcommand{\FOSTB}{\STB}

\newcommand{\FOSET}{\textup{SET}}
\newcommand{\FODIFF}{\textup{SYM-DIFF}}
\newcommand{\FOATMOSTK}{\textup{ATMOST}}
\newcommand{\FOSMALL}{\textup{$\sigma$-SMALL}}
\newcommand{\FOREPAIR}{\textup{$\sigma$-REPAIR}}
\newcommand{\FOADJUST}{\textup{$\sigma$-ADJUST}}
\newcommand{\FOCENTER}{\textup{$\sigma$-CENTER}}

\begin{document}

\title{The Complexity of Repairing, Adjusting, and Aggregating of Extensions in
  Abstract Argumentation}

\author{
Eun Jung Kim\\
\mbox{}\small LAMSADE, CNRS, Paris, France\\[-2pt]
\small eunjungkim78@gmail.com\\
\and
Sebastian Ordyniak\thanks{Research supported by Employment of Newly Graduated Doctors of
  Science for Scientific Excellence (CZ.1.07/\hskip0pt
  2.3.00/\hskip0pt 30.0009).}
\\[-2pt]
\mbox{}\small Department of Theoretical Computer Science, Masaryk University, Brno, Czech Republic\\[-2pt]
\small sordyniak@gmail.com\\
\and
Stefan Szeider\thanks{Research supported by the European Research
    Council, grant reference~239962 (COMPLEX REASON).}
\\[-2pt]
\mbox{}\small Institute of Information Systems, Vienna University of Technology,
Vienna, Austria\\[-2pt]
\small stefan@szeider.net
}

\date{}




\maketitle
\begin{abstract}
\noindent We study the computational complexity of problems that arise in
abstract argumentation in the context of dynamic argumentation,
minimal change, and aggregation.
In particular, we consider the following problems where always an
argumentation framework $F$ and a small positive integer $k$ are given.
\begin{itemize}
\item The \textsc{Repair} problem asks whether a given set of
  arguments can be modified into an extension by at most $k$
  elementary changes (i.e., the extension is of distance $k$ from the
  given set).
\item The \textsc{Adjust} problem asks whether a given extension can
  be modified by at most $k$ elementary changes into an extension that
  contains a specified argument.
\item The \textsc{Center} problem asks whether, given two extensions
  of distance $k$, whether there is a ``center''  extension that is a
  distance at most $k-1$ from both given extensions.
\end{itemize}
We study these problems in the framework of parameterized complexity,
and take the distance $k$ as the parameter. Our results covers several
different semantics, including admissible, complete, preferred, semi-stable and stable semantics.
\end{abstract}

\pagestyle{plain}
\thispagestyle{empty}
\section{Introduction}
Starting with the seminal work by Dung \cite{Dung95} the area of
argumentation has evolved to one of the most active research branches
within Artificial Intelligence
\cite{BenchcaponDunne07,RahwanSimari09}.  Dung's abstract
argumentation frameworks, where arguments are seen as abstract
entities which are just investigated with respect to how they relate
to each other, in terms of ``attacks'', are nowadays well understood
and different semantics (i.e., the selection of sets of arguments
which are jointly acceptable) have been proposed. Such sets of
arguments are called extensions of the underlying argumentation
framework.

Argumentation is an inherently dynamic process, and there has been
increasingly interest in the dynamic behavior of abstract
argumentation.  A first study in this direction was carried out by
Cayrol, et al.~\cite{CayrolDupinLagasquie08} and was concerned with
the impact of additional arguments on extensions.  Baumann and Brewka
\cite{BaumannBrewka10} investigated whether it is possible to modify a
given argumentation framework in such a way that a desired set of arguments becomes an
extension or a subset of an extension.  Baumann \cite{Baumann12}
further extended this line of research by considering the minimal
exchange necessary to enforce a desired set of arguments. In this
context, it is interesting to consider notions of \emph{distance}
between extensions. Booth~et
al.~\cite{BoothCaminadaPodlaszewskiRahwan12} suggested a general
framework for defining and studying distance measures.

A natural question that arises in the context of abstract
argumentation is how computationally difficult it is to decide whether
an argumentation framework admits an extension at all, or whether a given argument belongs
to at least one extension or to all extensions of the framework.  Indeed this
question has been investigated in a series of papers, and the exact
worst-case complexities have been determined for all popular semantics
\cite{CostemarquisDevredMarquis05,DimopoulosTorres96,Dung95,DunneBenchcapon02,CaminadaDunne08,DunneWooldridge09,DvorakWoltran10b}. Abstract
argumentation has also been studied in the framework of
\emph{parameterized complexity} \cite{DowneyFellows99} which admits a
more fine-grained complexity analysis that can take structural aspects
of the argumentation framework into account~\cite{Dunne07,DvorakOrdyniakSzeider12,KimOrdyniakSzeider11,DvorakSzeiderWoltran10,DvorakPichlerWoltran12}.

Surprisingly, very little is known on the computational complexity of
problems in abstract argumentation that arise in the context of
dynamic behavior of argumentation, such as finding an extension by
minimal change. However, as the distance in these problems are assumed
to be small, it suggests itself to consider the distance as the
parameter for a parameterized analysis.

\paragraph{New Contribution}
In this paper we provide a detailed complexity map of various problems
that arise in in the context of
dynamic behavior of argumentation.
In particular, we consider the following problems where always an
argumentation framework $F$ and a small positive integer $k$ are
given, and $\sigma$ denotes a semantics.
\begin{itemize}
\item The $\sigma$\hy \textsc{Repair} problem asks whether a given set of
  arguments can be modified into a $\sigma$\hy extension by at most $k$
  elementary changes (i.e., the extension is of distance $k$ from the
  given set).

  This problem is of relevance, for instance, when a $\sigma$\hy
  extension $E$ of an argumentation framework is given, and
  dynamically the argumentation framework changes (i.e., attacks are
  added or removed, new arguments are added). Now the set $E$ may not
  any more be a $\sigma$\hy extension of the new framework, and we
  want to repair it with minimal change to obtain a $\sigma$\hy
  extension.

\item The $\sigma$\hy \textsc{Adjust} problem asks whether a given
  $\sigma$\hy extension can be modified by at most $k$ elementary
  changes into a $\sigma$\hy extension that contains a specified
  argument.

  This problem is a variant of the previous problem, however, the
  argumentation framework does not change, but dynamically the
  necessity occurs to include a certain argument into the
  extension, by changing the given extension minimally.

\item The $\sigma$\hy \textsc{Center} problem asks whether, given two
  $\sigma$\hy extensions of distance $k$, whether there is a
  ``center'' $\sigma$\hy extension that is a distance at most $k-1$
  from both given extensions.

  This problem arises in scenarios of judgment aggregations, when,
  for instance, two extensions that reflect the opinion of two
  different agents are presented, and one tries to find a compromise
  extension that minimizes the distance to both extensions.

\end{itemize}
We study these problems in the framework of parameterized complexity,
and take the distance $k$ as the parameter. Our results covers several
different semantics, including admissible, complete, preferred,
semi-stable and stable semantics. The parameterized complexity of the
above problems are summarized in Figures~\ref{fig:comp-results}.

\begin{figure}
  \begin{center}
    \begin{tabular}{@{}l@{\qquad}@{\qquad}l@{\qquad}l@{}}
\toprule
      $\sigma$ & general & bounded degree  \\
\midrule
      $\ADM$ & $\W[1]$-hard & \FPT \\
      $\COM$ & $\W[1]$-hard & \FPT \\
      $\PRF$ & $\paracoNP$-hard & $\paracoNP$-hard  \\
      $\SEM$ & $\paracoNP$-hard & $\paracoNP$-hard \\
      $\STB$ & $\W[1]$-hard & \FPT\\ 
\bottomrule
    \end{tabular}
  \end{center}
  \caption{Parameterized Complexity of the problems \Repair{},
    \Adjust{}, and \Center{} for general argumentation frameworks
    and argumentation frameworks  of bounded
    degree, depending on the considered semantics.}
  \label{fig:comp-results}
\end{figure}

\section{Preliminaries}
\label{sec:preliminaries}

An \emph{abstract argumentation system} or \emph{argumentation
  framework} (\emph{AF}, for short) is a pair $(X,A)$ where $X$ is a
(possible infinite) set of elements called \emph{arguments} and
$A\subseteq X\times X$ is a binary relation called \emph{attack
  relation}. In this paper we will restrict ourselves to finite AFs,
i.e., to AFs for which $X$ is a finite set. If $(x,y)\in A$ we say that
\emph{$x$ attacks $y$} and that $x$ is an \emph{attacker} of~$y$.

An AF $F=(X,A)$ can be considered as a directed graph, and therefore it
is convenient to borrow notions and notation from graph theory.  For a
set of arguments $Y \subseteq X$ we denote by $F[Y]$ the AF $(Y,\SB
(x,y) \in A \SM x,y \in Y \SE)$ and by $F - Y$ the AF $F[X \setminus
Y]$. 


We define the \emph{degree} of an argument $x \in X$ to
be the number of arguments $y\in X\setminus \{x\}$ such that $(x,y)\in
A$ or $(y,x)\in A$. The maximum degree of an AF $F=(X,A)$ is the maximum
degree over all its atoms.
 We say a class $\CCC$ of AFs has bounded maximum
degree, or bounded degree for short, if there exists a constant $c$ such
that for every $F \in \CCC$ the maximum degree of the undirected graph
$\overline{F}$ is at most $c$.

If $E$ and $E'$ are $2$ sets of arguments of $F$ then we define $E
\bigtriangleup E'$ to be the symmetric difference between $E$ and
$E'$, i.e., $E \bigtriangleup E':=\SB x \in X \SM (x \in E \land x
\notin E') \lor (x \in E' \land x \notin E) \SE$. We also define
$\bsdist(E,E')$ to be $|E \bigtriangleup E'|$.
 
Let $F=(X,A)$ be an AF, $S \subseteq X$ and $x \in X$. We say that $x$
is \emph{defended} (in $F$) by $S$ if for each $x' \in X$ such that
$(x',x) \in A$ there is an $x'' \in S$ such that $(x'',x') \in A$. We
denote by $S_F^+$ the set of arguments $x \in X$ such that either $x \in
S$ or there is an $x' \in S$ with $(x',x) \in A$, and we omit the
subscript if $F$ is clear from the context. Note that in our setting
the set $S$ is contained in $S^+_F$.
We say $S$ is
\emph{conflict-free} if there are no arguments $x,x' \in S$ with $(x,x')
\in A$.

Next we define commonly used semantics of AFs, see the survey of
Baroni and Giacomin~\cite{BaroniGiacomin09}.  We consider a semantics $\sigma$ as a mapping
that assigns to each AF $F=(X,A)$ a family $\sigma(F) \subseteq 2^X$ of
sets of arguments, called \emph{extensions}.  We denote by \ADM{},
\COM{}, \PRF{}, \SEM{} and \STB{} the \emph{admissible}, \emph{complete},
\emph{preferred}, \emph{semi-stable} and \emph{stable} semantics, respectively.  These five semantics are characterized by the
following conditions which hold for each AF $F=(X,A)$ and each
conflict-free set $S\subseteq X$.
\begin{itemize}
\item $S \in \ADM{}(F)$ if each $s \in S$ is defended by
  $S$.
\item $S \in \COM{}(F)$ if $S \in \ADM{}(F)$ and every argument that is
  defended by $S$ is contained in $S$.


\item $S \in \PRF{}(F)$ if $S \in \ADM{}(F)$ and there is no $T \in
  \ADM{}(F)$ with $S \subsetneq T$.
\item
  $S \in \SEM{}(F)$ if $S \in \ADM{}(F)$ and there is no $T \in \ADM{}(F)$ with
  $S^+ \subsetneq T^+$.
\item $S \in \STB{}(F)$ if $S^+=X$.

\end{itemize}


\paragraph{Parameterized Complexity}
For our investigation we need to take two measurements into account:
the input size~$n$ of the given AF $F$ and the parameter~$k$ given as
the input to \Repair{}, \Adjust{}, and \Center{}. The theory of
\emph{parameterized complexity}, introduced and pioneered by Downey
and Fellows~\cite{DowneyFellows99}, provides the adequate concepts and
tools for such an investigation. We outline the basic notions of
parameterized complexity that are relevant for this paper, for an
in-depth treatment we refer to other sources
\cite{FlumGrohe06,Niedermeier06}.

An instance of a parameterized (decision) problem is a pair $(I,k)$
where $I$ is the \emph{main part} and $k$ is the \emph{parameter}; the
latter is usually a non-negative integer.  A parameterized problem is
\emph{fixed-parameter tractable} (FPT) if there exists a computable
function $f$ such that instances $(I,k)$ of size $n$ can be solved in
time $f(k)\cdot n^{O(1)}$, or equivalently, in \emph{fpt-time}.
Fixed-parameter tractable problems are also called \emph{uniform
  polynomial-time tractable} because if $k$ is considered constant,
then instances with parameter $k$ can be solved in polynomial time
where the order of the polynomial is independent of $k$, in contrast
to \emph{non-uniform polynomial-time} running times such as
$n^{O(k)}$.  Thus we have three complexity categories for
parameterized problems: (1) problems that are fixed-parameter
tractable (uniform polynomial-time tractable), (2) problems that are
non-uniform polynomial-time tractable, and (3) problems that are
$\NP$-hard or $\coNP$-hard if the parameter is fixed to some constant
(such as $k$-SAT which is $\NP$-hard for $k=3$). The major complexity
assumption in parameterized complexity is $\FPT \subsetneq
\W[1]$. Hence, $\W[1]$-hard problems are not fixed-parameter tractable
under this assumption. Such problems can still be non-uniform
polynomial-time tractable. Problems that fall into (3) above are said
to be \paraNP-hard or \paracoNP-hard.
The classes in parameterized
complexity are defined by {\em fpt-reduction}, which are many-one
reductions that can be computed in fpt-time, and where the parameter
of the target instance is bounded by a function of the parameter of
the source instance.

In our proofs of complexity results we will reduce from the following
problem, which is $\W[1]$\hy complete~\cite{Pietrzak03}.

\begin{myquote} \MCC{}

  \emph{Instance:} A natural number $k$, and a $k$-partite graph
  $G=(V,E)$ with partition $\{V_1,\dots,V_k\}$.

  \emph{Parameter:} $k$.

  \emph{Question:} Does $G$ contain a clique of size $k$?
\end{myquote}
W.l.o.g.~we may assume that the parameter $k$ of \MCC{} is even. To
see this, we reduce from \MCC{} to itself as follows. Given an
instance $(G,k)$ of \MCC{} we construct an equivalent instance
$(G',2k)$ of \MCC{} where $G'$ is obtained from the vertex-disjoint union of $2$ copies of
$G$ by adding all edges between the two copies. 

\section{Problems for Dynamic Argumentation}

In this section we present the problems that we consider for dynamic
argumentation.
Let $\sigma \in \{\ADM,\COM,\PRF,\SEM,\STB\}$. Recall that for two
sets $E$ and $E'$ of arguments $E \bigtriangleup E'$ and
$\bsdist(E,E')$ are defined as the symmetric difference and the
cardinality of the symmetric difference between $E$ and $E'$, respectively.
\begin{myquote}
  \Small{}

  \emph{Instance:} An AF $F=(X,A)$, a nonnegative integer $k$.

  \emph{Parameter:} $k$.

  \emph{Question:} Is there a nonempty extension $E\in \sigma(F)$ of size at most $k$?
\end{myquote}

\begin{myquote}
  \Repair{}

  \emph{Instance:} An AF $F=(X,A)$, a set of arguments $S\subseteq X$, a nonnegative integer $k$.

  \emph{Parameter:} $k$.

  \emph{Question:} Is there a nonempty extension $E\in \sigma(F)$ s.t.\ $\bsdist(E,S)\leq k$?
\end{myquote}

\begin{myquote}
  \Adjust{}

  \emph{Instance:} An AF $F=(X,A)$, an extension $E_0\in \sigma(F)$, an argument $t\in X$, a nonnegative integer~$k$.

  \emph{Parameter:} $k$.

  \emph{Question:} Is there an extension $E\in \sigma(F)$ s.t.\ $\bsdist(E,E_0)\leq k$ and $t\in E_0\bigtriangleup E$?
\end{myquote}

\begin{myquote}
  \Center{}

  \emph{Instance:} An AF $F=(X,A)$, two extensions $E_1, E_2 \in \sigma(F)$.

  \emph{Parameter:} $\bsdist(E_1,E_2)$.

  \emph{Question:} Is there an extension $E\in \sigma(F)$
  s.t.\ $\bsdist(E,E_i)< \bsdist(E_1,E_2)$ for every $i \in \{1,2\}$?
\end{myquote}

\section{Hardness Results}
\label{sec:hardness}

This section is devoted to our hardness results. We start by showing
that all the problems that we consider in the context of dynamic
argumentation are \W[1]-hard and hence unlikely to have FPT-algorithms.
\begin{THE}\label{the:W1-hard}
  Let $\sigma \in \{\ADM,\COM,\PRF,\SEM,\STB\}$. Then the problems
  \Small{}, \Repair{}, \Adjust{}, \Center{} are \W[1]-hard.
\end{THE}
Since the fpt-reductions used in the proof of
Theorem~\ref{the:W1-hard} can be computed in polynomial time, and
since the unparameterized version of \MCC{} is NP-hard, it follows
that the unparameterized versions of the four problems mentioned in
Theorem~\ref{the:W1-hard} are also NP-hard.
\fullproofversion{
We will have shown Theorem~\ref{the:W1-hard} after showing the
following $3$ Lemmas.
\begin{LEM}\label{lem:small-w1hard1} 
  Let $\sigma \in \{\ADM, \COM, \PRF, \SEM, \STB \}$. Then the problems \Small{} and
  \Repair{} are $\W[1]$-hard.
\end{LEM}
\begin{proof}
  We start by showing the lemma for the problem \Small{} by giving 
  an fpt-reduction from the \MCC{} problem to the
  \Small{} problem, when $\sigma$ is one of the listed semantics. Let
  $(G,k)$ be an instance of \MCC{} with partition $V_1,\dots,V_k$. We
  construct in fpt-time an AF $F$ such that there is an $E \in
  \sigma(F)$ with $|E|=k$ if and only if $G$ has a $k$-clique. The AF
  $F$ contains the following arguments: (1) $1$ argument $y_v$ for
  every $v \in V(G)$ and (2) for every $1 \leq i \leq k$, for every $v
  \in V_i$, and for every $1 \leq j \leq k$ with $j \neq i$, $1$
  argument $z_v^j$.

  For every $1 \leq i <j \leq k$, we denote by $Y[i]$ the set of
  arguments $\SB y_v \SM v \in V_i \SE$ and by $Z[i,j]$ the set of
  arguments $\SB z_v^j \SM v \in V_i \SE$. Furthermore, we set
  $Y:=\bigcup_{1 \leq i \leq k}Y[i]$ and $Z:=\bigcup_{1 \leq i < j
    \leq k}Z[i,j]$. For every $1 \leq i \leq k$, the AF $F$ contains
  the following attacks:
  \begin{itemize}
  \item $1$ attack from $y_v$ to $y_u$ for every $u,v \in Y[i]$ with
    $u \neq v$;
  \item $1$ self-attack for all arguments in $Z$;
  \item For every $v \in V_i$, $1$ attack from $z_v^j$ to $y_v$ for
    every $1 \leq j \leq k$ with $j\neq i$;
  \item For every $v \in V_i$, $1$ attack from $y_v$ to $z_u^j$
    for every $u \in V_i \setminus \{v\}$ and $1 \leq j \leq k$ with
    $j \neq i$.
  \item For every $\{u,v\} \in E(G)$ with $u \in V_i$ and $v \in V_j$, 
    $1$ attack from $y_u$ to $z_v^i$ and $1$ attack from $y_v$ to
    $z_u^j$.
  \end{itemize}
  This completes the construction of $F$. It remains to show that
  $G$ has a $k$-clique if and only if there is an $E \in \sigma(F)$
  with $|E|=k$. If $Q \subseteq V(G)$ we denote by $Y_Q$ the set of
  arguments $\SB y_q \SM q \in Q \SE$. We need the following claim.
  \begin{CLM}\label{clm:clkadm} A set $Q\subseteq V(G)$ is a $k$-clique
    in $G$ if and only if $Y_Q\in \ADM(F)$ and $Y_Q\neq \emptyset$.
  \end{CLM}
  Suppose that $Q \subseteq V(G)$ is a $k$-clique in $G$. Then $Y_Q$
  contains exactly $1$ argument from $Y[i]$ for every $1 \leq i \leq
  k$. Because there are no attacks between arguments in $Y[i]$ and
  $Y[j]$ for every $1 \leq i <j \leq k$ it follows that $Y_Q$ is
  conflict-free. To see that $Y_Q$ is also admissible let $y_v \in Y_Q
  \cap V_i$ and suppose that $y_v$ is attacked by an argument $x$ of
  $F$. It follows from the construction of $F$ that either $x \in
  Y[i]$ or $x \in \SB z_v^j \SM 1 \leq j \leq k \textup{ and }j\neq
  i\SE$. In the first case $x$ is attacked by $y_v$. In the second
  case $z_v^j$ is attacked by the argument in $Y[j] \cap Y_Q$ because
  $Q$ is a $k$-clique of $G$. Hence, $Y_Q \in \ADM(F)$ and $Y_Q \neq
  \emptyset$, as required.
  
  For the opposite direction, suppose that $E \in \ADM(F)$ and $E \neq
  \emptyset$. Because $E$ conflict-free it follows that $E \subseteq
  Y$ and $E$ contains at most $1$ argument from the set $Y[i]$
  for every $1 \leq i \leq k$. Because $E \neq \emptyset$ there is an
  argument $y_v \in Y[i] \cap E$. Because of the construction of $F$,
  $y_v$ is attacked by the arguments $\SB z_v^j \SM 1 \leq j \leq k
  \textup{ and } j \neq i \SE$. Hence, the arguments $\SB z_v^j \SM 1 \leq j \leq k
  \textup{ and } j \neq i \SE$ need to be attacked by arguments in
  $E$. However, the only arguments of $F$ that attack an argument
  $z_v^j$ with $j \neq i$ are the arguments $y_u \in Y[j]$ such that
  $\{u,v\} \in E(G)$. Hence, for every argument $y_v \in E \cap Y[i]$
  and every $1 \leq j \leq k$ with $j \neq i$
  there is an argument $y_u \in E \cap Y[j]$ such that $\{u,v\} \in
  E(G)$. It follows that the set $\SB v \SM y_v \in E \SE$ is a
  $k$-clique in $G$. This shows the claim.

  The previous claim shows that every non-empty admissible extension
  of $F$ corresponds to a $k$-clique of~$G$. It is now straightforward
  to check that every such extension is not only admissible but also
  complete, preferred, semi-stable, and stable. This shows the lemma
  for $\Small$. To show the Lemma for the \Repair{} problem 
  we note that $(F,\emptyset,k)$ is a
  \textsc{Yes}-instance for \Repair{} if and only if $(F,k)$ is a
  \textsc{Yes}-instance for \Small{}.
\end{proof}
\begin{LEM}\label{lem:adjust-w1hard} 
  Let $\sigma \in \{\ADM, \COM, \PRF, \SEM, \STB \}$. Then the problem \Adjust{} is
  $\W[1]$-hard.
\end{LEM}
\begin{proof} 
  We give an fpt-reduction from the \Small{} problem. Let $(F,k)$ be
  an instance of the \Small{} problem where $F=(X,A)$. We construct an
  equivalent instance $(F',E_1,E_2)$ of the \Adjust{} problem as
  follows. $F'=(X',A')$ is obtained from $F$ by adding $1$ argument
  $t$ and $2$ attacks $(t,x)$ and
  $(x,t)$ for every $x \in X$ to $F$. Because the argument $t$ attacks
  is attacked by all arguments in $X$ it follows that
  $\{t\}$ is a $\sigma$-extension of $F'$. In is now straightforward
  to show that $(F',\{t\},t,k+1)$ is a \textsc{Yes}-instance of
  \Adjust{} if and only if $(F,k)$ is a \textsc{Yes}-instance of
  \Small{}. This shows the lemma.
 \end{proof}
\begin{LEM}\label{lem:center-w1hard} 
  Let $\sigma \in \{\ADM, \COM, \PRF, \SEM, \STB\}$. Then the problem
  $\sigma$-Center is $\W[1]$-hard.
\end{LEM}
\begin{proof} 
  We give an fpt-reduction from the \Small{} problem. Let $(F,k)$ be
  an instance of the \Small{} problem where $F=(X,A)$. W.l.o.g. we can
  assume that $k$ is even. This follows from the remark in
  Section~\ref{sec:preliminaries} that \MCC{} is \W[1]-hard even if
  $k$ is even and the parameter preserving reduction from \MCC{} to
  \Small{} given in Lemma~\ref{lem:small-w1hard1}.
  We will construct an
  equivalent instance $(F',E_1,E_2)$ of the \Center{} problem as
  follows. $F'=(X',A')$ is obtained from $F$ by adding the following
  arguments and attacks to $F$.
  \begin{itemize}
  \item $2$ arguments $t$ and $t'$;
  \item the arguments in $W:=\{w_1,\dots,w_{k}\}$ and $W':=\{w_1',\dots,w_{k}'\}$;
  \item the arguments in $Z:=\{z_1,\dots,z_k\}$ and $Z':=\{z_1',\dots,z_k'\}$;
  \item attacks from $t$ to all arguments in $X \cup
    \{t'\} \cup Z \cup Z'$ 
    and attacks from $t'$ to all arguments in $X \cup \{t\} \cup Z
    \cup Z'$;
  \item attacks from $w_i$ to $\{t,w_i'\}$ and attacks from $w_i'$ to
    $\{t',w_i\}$ for every $1 \leq i \leq k$;
  \item self-attacks for the arguments $z_1,\dots,z_k$ and $z_1',\dots,z_k'$;
  \item attacks from $z_i$ to $\{w_i,w_i'\}$ and from $X$ to $z_i$ for
    every $1 \leq i \leq k$;
  \item attacks from $\{w_i,w_i'\}$ to $z_i'$ and from $z_i'$ to $X$
    for every $1 \leq i \leq k$;
  \end{itemize}
  We set $E_0:=\{w_1,\dots,w_{k/2},w_{k/2+1}',\dots,w_k'\}$,
  $E_1:=\{t\}\cup W'$, $E_2:=\{t'\}\cup W$, and
  $k':=\bsdist(E_1,E_2)-1=2(k+1)-1=2k+1$. Then $E_1$ and
  $E_2$ are $\sigma$-extensions and hence $(F',E_1,E_2)$ is a valid
  instance of the \Center{} problem. It remains to show that $(F,k)$
  is a \textsc{Yes} instance of \Small{} if and only if $(F',E_1,E_2)$
  is a \textsc{Yes} instance of \Center{}.

  Suppose that $(F,k)$ is a \textsc{Yes} instance of \Small{} and let
  $E$ be a non-empty $\sigma$-extension of cardinality at most $k$ witnessing
  this. Then $E':=E \cup E_0$ is a $\sigma$-extension of $F'$ and
  $\bsdist(E',E_i)=k+k+1=2k+1\leq k'$ for $i \in \{1,2\}$, as
  required.

  For the reverse direction suppose that $E'$ is a $\sigma$-extension
  of $F'$ with $\bsdist(E',E_i)\leq k'$ for $i \in \{1,2\}$. We need
  the following claim.
  \begin{CLM}
    $E'$ does not contain $t$ or $t'$.
  \end{CLM}
  Suppose for a contradiction that $E'$ contains one of $t$ and
  $t'$. Because $t$ and $t'$ attack each other $E'$ cannot contain
  both $t$ and $t'$. W.l.o.g. we can assume that $t \in E'$. Because $E'$
  is a $\sigma$-extension $E'$ is also admissible. Since, 
  the arguments $w_1,\dots,w_k$ attack $t$, there need to be arguments
  in $E'$ that attack these arguments. It follows that $E'$ contains
  the arguments $w_1',\dots,w_k'$. But then
  $\bsdist(E',E_2)\geq\bsdist(E_1,E_2)$ a contradiction.
  \begin{CLM}
    $E' \cap X$ is a non-empty $\sigma$-extension of $F$ and
    $E'$ contains exactly one of the arguments $w_i$ and $w_i'$ for
    every $1 \leq i \leq k$.
  \end{CLM}
  It follows from the previous claim that $E'$ does not contain $t$ or
  $t'$. Furthermore, because of the self-loops of the arguments in $Z
  \cup Z'$, $E'$ contains only arguments from $X \cup W \cup
  W'$. Since the arguments in $X$ do not attack or are
  attacked by arguments in $W \cup W'$ it follows that $E' \cap X$ is
  a $\sigma$-extension of $F$. To see that $E' \cap X$ is also not
  empty, suppose for a contradiction that this is not the case.
  Then because $E'$ is non-empty, $E'$
  has to contain at least $1$ argument from $W \cup W'$. However, any
  argument in $W \cup W'$ is attacked by an argument in $Z$ and the
  only arguments that attack arguments in $Z$ are the arguments in $X
  \cup \{t,t'\}$. Again using the previous claim and the
  fact that $E'$ is admissible, it follows that $E'$ has to contain at
  least $1$ argument from $X$, as required. 
  It remains to show that
  $E'$ contains exactly one of $w_i$ and $w_i'$ for every $1 \leq i
  \leq k$. Because $E'$ contains at least $1$ argument from $X$ and
  all arguments in $X$ are attacked by all arguments in $Z'$, $E'$
  needs to contain arguments that attack all arguments in $Z'$. However,
  the only arguments that attack arguments in $Z'$ are the arguments
  in $\{t,t'\}\cup W \cup W'$. Using the previous claim it follows
  that the only way for $E'$ to attack all arguments in $Z'$ is to
  contain at least $1$ of $w_i$ and $w_i'$ for every $1 \leq i \leq
  k$. The claim now follows by observing that because $E'$ is
  conflict-free, it cannot contain both arguments $w_i$ and $w_i'$ for
  any $1 \leq i \leq k$. This proves the claim.

  Since $E'$ contains exactly $1$ of $w_i$ and $w_i'$ for every $1
  \leq i \leq k$ we obtain that either $|W\setminus E'|\geq k/2$ or
  $|W'\setminus E'|\geq k/2$. W.l.o.g. we can assume that $|W\setminus
  E'|\geq k/2$. But then $\bsdist(E',E_2)=|E' \cap X|+1+2|W  \setminus
  E'|=|E' \cap X|+k+1$ and because $\bsdist(E',E_2)\leq k'=2k+1$ it
  follows that $|E' \cap X|\leq k$. This concludes the proof of the
  lemma.
\end{proof}
This concludes the proof of Theorem~\ref{the:W1-hard}.
}

\bigskip\noindent
In the next section we will show that, when considering AFs of bounded
maximum degree, then fixed-parameter tractability can be obtained for the
admissible, complete, and stable semantics. Unfortunately, this
positive result does not hold for the preferred and semi-stable
semantics as the following result shows.
\begin{THE}\label{the:coNP-hard}
  Let $\sigma \in \{\PRF,\SEM\}$. Then the problems
  \Small{}, \Repair{}, \Adjust{}, \Center{} are $\paracoNP$-hard, even
  for AFs of maximum degree $5$.
\end{THE}

The remainder of this section is devoted to the proof of Theorem~\ref{the:coNP-hard}.
\begin{LEM}\label{the:small-paracoNP} 
  Let $\sigma \in \{\PRF, \SEM\}$. Then the problems \Small{} and \Repair{}
  are \paracoNP-hard (for parameter equal to $1$), even for AFs of maximum
  degree at most $5$.
\end{LEM}
\begin{proof} 
  We will show the theorem by providing a polynomial reduction from
  the \textsc{$3$-CNF-$2$-UnSatisfiablily} problem which is well-known to be
  $\coNP$-hard~\cite{GareyJohnson79}. The
  \textsc{$3$-CNF-$2$-UnSatisfiablily} problem ask whether a given
  $3$-CNF-$2$ formula $\Phi$, i.e., $\Phi$ is a CNF formula where
  every clause contains at most $3$ literals and every literal occurs in at most $2$
  clauses, is not satisfiable. Let $\Phi$ be a such a $3$-CNF-$2$
  formula with clauses $C_1,\dots,C_m$ and variables $x_1,\dots,x_n$. 
  We will (in polynomial time) construct an AF $F=(X,A)$ such that (1)
  $F$ has degree at most $5$ and (2) $\Phi$ is not
  satisfiable if and only if there is an $E \in \sigma(F)$ with
  $|E|=1$. This implies the theorem.

  $F$ contains the following arguments: (1) $2$ arguments $\Phi$ and
  $\overline{\Phi}$, (2) $1$ argument $C_j$ for every $1 \leq j \leq m$,
  (3) $2$ arguments $x_i$ and $\overline{x_i}$ for every $1 \leq i \leq n$,
  and (4) $1$ argument $e$. Furthermore, $F$ contains the following
  attacks: (1) $1$ self-attack for the arguments
  $\overline{\Phi}$ and $C_1,\dots,C_m$, (2) $1$ attack from $\Phi$ to
  $\overline{\Phi}$, (3) $1$ attack from $C_j$ to $\Phi$ for every $1 \leq
  j \leq m$, (4) $1$ attack from $x_i$ to $C_j$ for every $1 \leq i
  \leq n$ and $1 \leq j \leq m$ such that $x_i \in C_j$, (5) $1$
  attack from $\overline{x_i}$ to $C_j$ for every $1 \leq i
  \leq n$ and $1 \leq j \leq m$ such that $\overline{x_i} \in C_j$, (6) $2$
  attacks from $x_i$ to $\overline{x_i}$ and from $\overline{x_i}$ to $x_i$ for
  every $1 \leq i \leq n$, and (7) $2$ attacks from $\overline{\Phi}$ to
  $x_i$ and to $\overline{x_i}$ for every $1 \leq i \leq n$.

  Note that the constructed AF $F$ does not have bounded
  degree. Whereas all arguments in $X \setminus \{\Phi,\overline{\Phi}\}$
  have degree at most $5$, the degree of the arguments $\Phi$ and
  $\overline{\Phi}$ can be unbounded. However, the following simple trick can be used
  to transform $F$ into an AF with bounded degree. 

  Let $B(i)$ be an undirected rooted binary tree with
  root $r$ and $i$ leaves $l_1,\dots,l_i$ and let $B'(i)$ be obtained from
  $B(i)$ after subdividing every
  edge of $B(i)$ once, i.e., every edge $\{u,v\}$ is replaced with $2$
  edges $\{u,n_{uv}\}$ and $\{n_{uv},v\}$ where $n_{uv}$ is a new
  vertex for every such edge. We denote by $B(\Phi)$ the rooted
  directed tree obtained from $B'(m)$ after directing every edge of $B'(m)$
  towards the root $r$ and introducing a self-attack for every vertex
  in $V(B'(m)) \setminus V(B(m))$, i.e., all vertices introduced for
  subdividing edges of $B(m)$ are self-attacking in $B(\Phi)$. Then to
  ensure that the
  argument $\Phi$ has bounded degree in $F$ we first delete the attacks from
  the arguments $C_1,\dots,C_m$ to $\Phi$ in $F$. We then add a copy of
  $B(\Phi)$ to $F$ and identify $\Phi$ with the root~$r$. Finally, we
  add $1$ attack from $C_j$ to $l_j$ for every $1 \leq j \leq m$. 
  Observe that this construction maintains the property of $F$ that if
  a $\sigma$-extension of $F$ contains $\Phi$ then it also has to
  contain at least $1$ attacker of every argument $C_1,\dots,C_m$.
  
  Let $B(\overline{\Phi})$ be the rooted directed tree obtained from
  $B'(2n)$ after directing every edge of $B'(2n)$
  away from the root $r$ and introducing a self-attack for every vertex
  in $V(B(2n))$. To ensure that also the argument $\overline{\Phi}$ has bounded
  degree we first delete the attacks from
  the argument $\overline{\Phi}$ to $x_1,\overline{x_1},\dots,x_n,\overline{x_n}$ in
  $F$. We then add a copy of
  $B(\overline{\Phi})$ to $F$ and identify $\overline{\Phi}$ with the root $r$. Finally, we
  add $2$ attacks from $l_{i}$ to $x_i$ and from $l_{n+i}$ to
  $\overline{x_i}$ for every $1\leq i \leq n$. 
  Observe that this construction maintains the property of $F$ that if
  a $\sigma$-extension of $F$ contains $x_i$ or $\overline{x_i}$ for some 
  $1 \leq i \leq n$ then $\overline{\Phi}$ needs to be attacked by the
  argument $\Phi$ in $F$ and hence such a $\sigma$-extension has to
  contain the argument $\Phi$.

  Clearly, after applying the above transformations to $F$ the
  resulting AF has maximum degree at most $5$. However, to make the
  remaining part of the proof less technical we will give the proof
  only for the AF $F$. We will need the
  following claim.
  \begin{CLM}\label{clm:bdeg-small-repair-1}\sloppypar
    If there is an $E \in \ADM(F)$ that contains at least $1$ argument
    in $\{\Phi,x_1,\overline{x_2},\dots,x_n,\overline{x_n}\}$ then $\Phi \in E$.
  \end{CLM}
  Let $E \in \ADM(F)$ with $E \cap \{\Phi,x_1,\overline{x_2},\dots,x_n,\overline{x_n}\}\neq
  \emptyset$. If $\Phi \in E$ then the claim holds. So suppose
  that $\Phi \notin E$. Then there is an $1 \leq i \leq n$ such that
  either $x_i \in E$ or $\overline{x_i} \in E$. Because both $x_i$ and
  $\overline{x_i}$ are attacked
  by the argument $\overline{\Phi}$ and the only argument (apart from
  $\overline{\Phi}$) that attacks $\overline{\Phi}$ in $F$ is $\Phi$ it follows
  that $\Phi \in E$. This shows the claim.

  \begin{CLM}\label{clm:bdeg-small-repair-2}\sloppypar
    There is an $E \in \ADM(F)$ that contains at least $1$ argument
    in $\{\Phi,x_1,\overline{x_2},\dots,x_n,\overline{x_n}\}$ if and only if the
    formula $\Phi$ is satisfiable.
  \end{CLM}
  Suppose there is an $E \in \ADM(F)$ with $E \cap
  \{\Phi,x_1,\dots,x_n\} \neq \emptyset$. Because of the previous
  claim we have that $\Phi \in E$. Because $\Phi \in E$ and $\Phi$ is
  attacked by the arguments $C_1,\dots,C_m$ it follows that the
  arguments $C_1,\dots,C_m$ must be attacked by some argument in
  $E$. Let $a(C_j)$ be an argument in $E$ that attacks $C_j$. Then
  $a(C_j)$ is an argument that corresponds to a literal of the clause
  $C_j$. Furthermore, because $E$ is conflict-free the set $L:=\SB a(C_j)
  \SM 1 \leq j \leq m \SE$ does not contain arguments that correspond
  to complementary literals. Hence, $L$ corresponds to a satisfying
  assignment of $\Phi$.

  For the reverse direction suppose $\Phi$ is satisfiable and let $L$
  be a set of literals witnessing this, i.e., $L$ is a set of literals
  that correspond to a satisfying assignment of $\Phi$. It is
  straightforward to check that $E:=\{\Phi\} \cup L$ is in $\ADM(F)$.
  This completes the proof of the claim.

  \begin{CLM}\label{clm:bdeg-small-repair-3}
    Let $E \in \sigma(F)$. Then $e \in E$.
  \end{CLM}
  This follows directly from our assumption that $\sigma \in
  \{\PRF,\SEM\}$ and the fact that the argument $e$ is isolated in
  $F$.

  We are now ready to show that $\Phi$ is not satisfiable if and only
  if there is an $E \in \sigma(F)$ with $|E|=1$. 
  So suppose that
  $\Phi$ is not satisfiable. It follows from the previous claim that
  $E \cap \{\Phi,x_1,\overline{x_1},\dots,x_n,\overline{x_n}\}=\emptyset$ for
  every $E \in \ADM(F)$ and hence also for every $E \in
  \sigma(F)$. Because of the self-attacks of the arguments
  in $\{\overline{\Phi},C_1,\dots,C_m\}$, we obtain that $E \subseteq
  \{e\}$. Using the previous claim, we have $E=\{e\}$ as required.

  For the reverse direction suppose that there is an $E \in \sigma(F)$
  with $|E|=1$. Because of the previous claim it follows that
  $E=\{e\}$. Furthermore, because of the maximality condition of the preferred and
  semi-stable semantics it follows that there is no $E \in \ADM(F)$
  such that $E \cap \{\Phi,x_1,\overline{x_1},\dots,x_n,\overline{x_n}\} \neq
  \emptyset$ and hence (using Claim~\ref{clm:bdeg-small-repair-2}) the
  formula $\Phi$ is not satisfiable.
\end{proof}

\begin{LEM}\label{lem:adjust-paracoNP} 
  Let $\sigma \in \{\PRF, \SEM\}$. Then the problem \Adjust{} is
  \paracoNP-hard (for parameter equal to $2$) even if the maximum
  degree of the AF is bounded by $5$.
\end{LEM}
\begin{proof} 
  We use a similar construction as in the proof of
  Theorem~\ref{the:small-paracoNP}. Let $F$ be the AF constructed from
  the $3$-CNF-$2$ formulas $\Phi$ as in the proof of
  Theorem~\ref{the:small-paracoNP}. Furthermore, let $F'$ be the AF
  obtained from $F$ after removing the argument $e$ and adding $4$
  novel arguments $t_1$, $t_1'$, $t_2$, and $t_2'$ and the attacks
  $(t_1,\Phi)$, $(\Phi,t_1)$, $(t_1,t_2)$, $(t_2,t_1)$, $(t_1,t_1')$,
  $(t_2,t_2')$, $(t_1',t_1')$, and $(t_2',t_2')$ to $F$. 
  Because $F$ has degree bounded by $5$ (and the degree of the
  argument $\Phi$ in $F$ is $3$) it follows that the maximum
  degree of $F'$ is $5$ as required. We claim that $(F',\{t_1\},t_1,2)$
  is a \textsc{Yes}-instance of \Adjust{} if and only if $\Phi$ is not
  satisfiable.

  It is straightforward to verify that the Claims~\ref{clm:bdeg-small-repair-1}
  and~\ref{clm:bdeg-small-repair-2} also hold for the AF $F'$.
  We need the following additional claims.
  \begin{CLM}\label{clm:bdeg-adjust-1}
    $\{t_1\} \in \sigma(F')$. 
  \end{CLM}
  Clearly, $\{t_1\} \in  \ADM(F')$. We first show that for every $E
  \in \ADM(F')$ with $t_1\in E$ it holds that $E=\{t_1\}$.
  Let $E \in \ADM(F')$ with $t_1 \in E$. Because of
  the attacks between $t_1$ and $t_2$ and
  between $t_1$ and $\Phi$ it follows that $\Phi,t_2 \notin
  E$. Using Claim~\ref{clm:bdeg-small-repair-1} it follows that 
  also none of the arguments in
  $\{x_1,\overline{x_1},\dots,x_n,\overline{x_n}\}$ are contained in $E$.
  Furthermore, because of the self-attacks in $F'$ it also holds
  that none of the arguments in
  $\{\overline{\Phi},C_1,\dots,C_m,t_1',t_2'\}$ are contained in $E$. 
  Hence, $E=\{t_1\}$, as required. This implies that $\{t_1\} \in
  \PRF(F')$. To show that $\{t_1\} \in \SEM(F')$ observe that 
  $t_1$ is the only argument in $F$ (apart from $t_1'$ itself) that
  attacks $t_1'$. Furthermore, because $t_1'$ attacks itself it cannot
  be in any semi-stable extension of $F'$. Hence, $\{t_1\} \in
  \SEM(F')$. This shows the claim.
  \begin{CLM}\label{clm:bdeg-adjust-2}
    $\{t_2\} \in \sigma(F')$ if and only if $\Phi$ is not satisfiable. 
  \end{CLM}
  \sloppypar Suppose that $\{t_2\} \in \sigma(F')$. If $\{t_2\} \in \PRF(F')$
  then there is no $E \in \ADM(F')$ with $\{t_2\} \subsetneq E$. It
  follows that there is no $E' \in \ADM(F')$ with $E' \cap
  \{\Phi,x_1,\overline{x_1},\dots,x_n,\overline{x_n}\} \neq \emptyset$, since
  such an $E'$ could be added to $E$. Using
  Claim~\ref{clm:bdeg-small-repair-2} it follows that $\Phi$ is not
  satisfiable. If on the other hand $\{t_2\} \in \SEM(F')$ then
  because $t_2$ is the only argument that attacks $t_2'$ and because
  of the self-attack of $t_2'$ it follows again that there is no $E \in
  \ADM(F')$ with $\{t_2\} \subsetneq E$. Hence, using the same
  arguments as for the case $\{t_2\} \in \PRF(F')$ we again obtain
  that $\Phi$ is not satisfiable. 

  For the reverse direction suppose that $\Phi$ is not
  satisfiable. Because of Claim~\ref{clm:bdeg-small-repair-2} we
  obtain that every $E \in \ADM(F')$ (and hence also every $E
  \in \sigma(F')$) contains no argument in
  $\{\Phi,x_1,\overline{x_1},\dots,x_n,\overline{x_n}\}$. Because $\{t_2\} \in
  \ADM(F')$ and the argument $t_2$ attacks the only remaining argument $t_1$ with no
  self-attack it follows that $\{t_2\} \in \sigma(F')$.

  To show the theorem it remains to show that there is an $E' \in
  \sigma(F')$ with $t_1 \notin E'$ and $\bsdist(E,E')\leq 2$ if and only
  if the formula $\Phi$ is not satisfiable. First observe that
  because of
  Claim~\ref{clm:bdeg-adjust-1}, $\emptyset \notin \sigma(F')$ and
  hence $E'$ must contain exactly $1$ argument other than $t_1$.
  Consequently, it remains to show that there is an argument $x \in X
  \setminus \{t_1\}$ such that $\{x\} \in \sigma(F')$ if and only if
  $\Phi$ is not satisfiable.

  Suppose that there is an $x \in X\setminus \{t_1\}$ with $\{x\} \in
  \sigma(F')$. If $x\in \{\Phi,x_1,\overline{x_1},\dots,x_n,\overline{x_n}\}$
  then because of Claim~\ref{clm:bdeg-small-repair-1} it holds that
  $x=\Phi$. However, assuming that $\Phi$ contains at least $1$ clause
  it follows that $\{x\}$ is not admissible, and hence $x \neq \Phi$.
  Considering the self-attacks of $F$ we obtain that $x=t_2$. Hence,
  the forward direction follows from Claim~\ref{clm:bdeg-adjust-2}.

  The reverse direction follows immediately from
  Claim~\ref{clm:bdeg-adjust-2}. This concludes the proof of the
  theorem.
\end{proof}

\begin{LEM}\label{lem:center-paracoNP} 
  Let $\sigma \in \{\PRF, \SEM\}$. Then the problem \Center{} is
  \paracoNP-hard (for parameter equal to $6$) even if the maximum
  degree of the AF is bounded by $5$.
\end{LEM}
\begin{proof} \sloppypar
  We use a similar construction as in the proof of
  Theorem~\ref{the:small-paracoNP}. Let $F$ be the AF constructed from
  the $3$-CNF-$2$ formulas $\Phi$ as in the proof of
  Theorem~\ref{the:small-paracoNP}. Furthermore, let $F'$ be the AF
  obtained from $F$ after removing the argument $e$ and adding $12$
  novel arguments $t$, $t'$, $w_1$, $w_2$, $w_1'$, $w_2'$, $z$, $z'$,
  $z_1$, $z_1'$, $z_2$, $z_2'$ and the attacks
  $(t,z)$, $(z,z)$, $(t',z')$, $(z',z')$, $(w_1,z_1)$, $(z_1,z_1)$, $(w_1',z_1')$,
  $(z_1',z_1')$, $(w_2,z_2)$, $(z_2,z_2)$, $(w_2',z_2')$,
  $(z_2',z_2')$,
  $(t,\Phi)$, $(\Phi,t)$, $(t',\Phi)$, $(\Phi,t')$, $(t,t')$,
  $(t',t)$, $(w_1,w_1')$, $(w_1',w_1)$, $(w_2,w_2')$, $(w_2',w_2)$,
  $(w_1,t)$, $(w_2,t)$, $(w_1',t')$, and $(w_2',t)$ to $F$.
  Because $F$ has degree bounded by $5$ (and the degree of the
  argument $\Phi$ of $F$ is $3$) it follows that the maximum
  degree of $F'$ is $5$ as required. We claim that $(F',\{t,w_1',w_2'\},\{t',w_1,w_2\})$
  is a \textsc{Yes}-instance of \Center{} if and only if $\Phi$ is not
  satisfiable.

  It is straightforward to verify that the Claims~\ref{clm:bdeg-small-repair-1}
  and~\ref{clm:bdeg-small-repair-2} also hold for the AF $F'$.
  We need the following additional claims.
  \begin{CLM}\label{clm:bdeg-center-1}
    $\{t,w_1',w_2'\} \in \sigma(F')$ and $\{t',w_1,w_2\}\in \sigma(F')$.
  \end{CLM}
  We show that $\{t,w_1',w_2'\} \in \sigma(F')$. The case for
  $\{t',w_1,w_2\}\in \sigma(F')$ is analogous due to the symmetry of~$F'$. Clearly, $\{t,w_1',w_2'\} \in \ADM(F')$.

  We first show that for every $E
  \in \ADM(F')$ with $t \in E$ it holds that $E=\{t,w_1',w_2'\}$.
  Let $E \in \ADM(F')$ with $t \in E$. Clearly,
  $E$ does not contain $\Phi$, $t'$, $w_1$ or $w_2$ (since these
  arguments are neighbors of $t$ in $F'$). Using
  Claim~\ref{clm:bdeg-small-repair-1} it follows that 
  also none of the arguments in
  $\{x_1,\overline{x_1},\dots,x_n,\overline{x_n}\}$ are contained in $E$.
  Furthermore, because of the self-attacks in $F'$ it also holds
  that none of the arguments in
  $\{\overline{\Phi},C_1,\dots,C_m,z,z',z_1,z_1',z_2,z_2'\}$ are contained in $E$. 
  Hence, $E \subseteq \{t_1,w_1',w_2'\}$. However, because $t$ is
  attacked by $w_1$ and $w_2$ in $F$ and $w_1'$ and $w_2'$ are the only
  arguments of $F'$ that attack $w_1$ and $w_2$ it follows that
  $E=\{t,w_1',w_2'\}$. This implies that $\{t,w_1',w_2'\} \in
  \PRF(F')$. To show that $\{t,w_1',w_2'\} \in \SEM(F')$ observe that 
  $t$ is the only argument in $F'$ (apart from $z$ itself) that
  attacks $z$. Furthermore, because $z$ attacks itself it cannot
  be in any semi-stable extension of $F'$. Hence, $\{t,w_1',w_2'\} \in
  \SEM(F')$. This shows the claim.

  The proof of the previous claim actually showed the following
  slightly stronger statement.
  \begin{CLM}\label{clm:bdeg-center-2}
    Let $E \in \sigma(F')$ with $t \in E$. Then $E=\{t,w_1',w_2'\}$.
    Similarly, if $E \in \sigma(F')$ with $t' \in E$. Then
    $E=\{t',w_1,w_2\}$.
  \end{CLM}
  We are now ready to show that there is an $E \in \sigma(F')$ with
  $\bsdist(E,E_i)<\bsdist(E_1,E_2)=6$ for every $i \in \{1,2\}$ if and
  only if the formula $\Phi$ is not satisfiable. 
  
  Suppose that there is an $E \in \sigma(F')$ with
  $\bsdist(E,E_i)<\bsdist(E_1,E_2)=6$ for every $i \in \{1,2\}$. Then
  because of Claim~\ref{clm:bdeg-center-2} $E$ does not contain $t$ or
  $t'$. If there is an $E \in \sigma(F')$ with $\Phi \in E$ then we
  can assume (because of the maximality properties of the two
  semantics) that $E$ contains $1$ of $x_i$ or $\overline{x_i}$ for every
  $1 \leq i \leq n$. Hence, if $\Phi \in E$ and the formula $\Phi$
  contains at least $5$ variables (which we can assume w.l.o.g.) then
  $\bsdist(E,E_1)>5$. Consequently, $\Phi \notin E$ and it follows
  from Claims~\ref{clm:bdeg-small-repair-1}
  and~\ref{clm:bdeg-small-repair-2} that $\Phi$ is not satisfiable, as
  required.

  For the reverse direction suppose that $\Phi$ is not satisfiable. 
  Let $E:=\{w_1,w_2'\}$. Clearly, $\bsdist(E,E_i)=3<5$, as
  required. It remains to show that $E \in \sigma(F')$. It is easy to
  see that $E \in \ADM(F')$. Furthermore, because
  $\Phi$ is not satisfiable it follows from
  Claim~\ref{clm:bdeg-small-repair-2} that no $E' \in \sigma(F')$ can contain
  an argument in $\{\Phi,x_1,\overline{x_1},\dots,x_n,\overline{x_n}\}$ and
  hence $E \in \PRF(F')$. The maximality of $E$ with respect to the
  semi-stable extension now follows from the fact that $w_1$ and
  $w_2'$ are the only arguments that attack the arguments $z_1$ and
  $z_2'$ and because of their self-attacks none of $z_1$ and $z_2$ can
  them-self be contained in a semi-stable extension. This completes the
  proof of the theorem.
\end{proof}
Lemmas~\ref{the:small-paracoNP},~\ref{lem:adjust-paracoNP},
and~\ref{lem:center-paracoNP} together imply Theorem~\ref{the:coNP-hard}.

\section{Tractability Results}

Unfortunately, the results of the previous section draw a rather
negative picture of the complexity of problems important to dynamic
argumentation. In particular, Theorem~\ref{the:coNP-hard} strongly
suggests that at least for the preferred and semi-stable semantics
these problems remain intractable even when the degree of arguments is
bounded by a small constant. The hardness of these problems under the
preferred and semi-stable semantics seems to originate from their
maximality conditions. In this section we take a closer look at the
complexity of our problems for the three remaining semantics, i.e.,
the admissible, complete, and stable semantics. We show that in
contrast to the preferred and semi-stable semantics all our problems
become fixed-parameter tractable when the arguments of the given AF
have small degree.  In particular, we will show the following result.
\begin{THE}\label{the:tract}
  Let $\sigma \in \{\ADM,\COM,\STB\}$ and $c$ a natural number. Then the problems \Small{},
  \Repair{}, \Adjust{}, and \Center{} are fixed-parameter tractable if
  the maximum degree of the input AF is bounded by $c$.
\end{THE}
To show the above theorem we will reduce it to a Model Checking
Problem for First Order Logic. For a class $\SSS$ of  finite relational
structures we consider the following parameterized problem.
\begin{myquote}
  \textsc{$\SSS$-FO Model Checking}

  \emph{Instance:} A finite structure $S$ with $S \in \SSS$ and a
  First Order (FO) formula $\varphi$.

  \emph{Parameter:} $|\varphi|$ (i.e., the length of $\phi$).

  \emph{Question:} Does $S$ satisfy (or model) $\varphi$, i.e., is $S
  \models \varphi$?
\end{myquote}
For a formal definition of the syntax and semantics of FOL and
associated notions we refer the reader to a standard text~\cite{FlumGrohe06}. Central
to our result is the following proposition.
\begin{PRO}[\cite{Seese96}]\label{pro:seese}
  Let $\SSS$ be a class of structures whose maximum degree is bounded
  by some constant. Then the problem
  \textsc{$\CCC$-FO Model Checking} is
  fixed-parameter tractable.
\end{PRO}
We note here that we define the maximum degree of a structure $S$ in terms
of the maximum degree of its associated Gaifman graph, which is the
undirected graph whose vertex set is the universe of $S$,
and where two vertices are joined by and edge if they appear together
in a tuple of a relation of $S$.

There exists several extensions of the above result to even more general
classes, e.g., the class of graphs with locally bounded treewidth. Due
to the technicality of the definition of these classes we refrain
from stating these results in detail and refer the interested reader
to~\cite{Kreutzer09}. Results such as the one above are also commonly
refereed to as meta-theorems, i.e., they allow us to make statements
about a wide variety of algorithmic problems. Similar meta-theorems have
been used before in the context of Abstract Argumentation (see,
e.g.,~\cite{Dunne07,KimOrdyniakSzeider11,DvorakSzeiderWoltran12}).

We will now show how to reduce our problems to the
\textsc{$\SSS$-FO Model Checking} problem. To do so we need to (1)
represent the input of \Small{}, \Repair{}, \Adjust{}, \Center{} in terms
of finite structures (whose maximum degree is bounded in terms of the
maximum degree of the input AF), and (2) give a FO sentence that is
satisfied by the structure obtained in step (1) if and only if the
given instance of \Small{}, \Repair{}, \Adjust{}, \Center{} is
a \textsc{Yes} instance.

We start by defining the structures that correspond to the input of
our problems. For all of our problems, the structure has 
universe $X$ and one binary relation $A$ that is equal to the attack
relation of the AF $F=(X,A)$, which is given in the input. Additionally,
the resulting structures will contain
unary relations, which represent arguments or sets of arguments,
respectively, which are given in the input. For instance, the structure for an
instance $(F,E_0,t,k)$ of \Adjust{} has universe $X$, one binary
relation $A$ that equals the attack relation of $F$, one unary
relation $E_0$ that equals the set $E_0$, and one unary relation $T$
with $T:=\{t\}$. The structures for the problems \Small{}, \Repair{}, and
\Center{} are defined analogously. It is straightforward to verify that
the maximum degree of the structures obtained in this way is equal to
the maximum degree of the input AF.


Towards defining the FO formulas for step (2) we start by defining the
following auxiliary formulas. Due to the complexity of the FO formulas
that we need to define,
we will introduce some additional notation that will allow us to reuse formulas by
substituting parts of other formulas. We will provide examples how to
interpret the notation when these formulas are introduced.

In the following let $l$ be a natural
number, and let $\phi(x)$, $\phi_1(x)$,
and $\phi_2(x)$ be FO formulas with free variable $x$.

\sloppypar The formula $\FOSET[l](x_1,\dots,x_l,y)$ is satisfied if and
only if the argument $y$ is equal to at least $1$ of the arguments $x_1,\dots,x_l$.
\begin{quote}
  $\FOSET[l](x_1,\dots,x_l,y):= (y=x_1 \lor \dots \lor y=x_l$)
\end{quote}
We note here that the notation $\FOSET[l]$ means that the exact definition
of the formula $\FOSET[l]$ depends on the value of $l$, e.g., if $l=3$
then $\FOSET[l]$ is the formula $y=x_1 \lor y=x_2 \lor y=x_3$.

\sloppypar The formula $\FOIS[\phi(x)]$ is satisfied if and
only if the set of arguments that satisfy the formula $\phi(x)$ is
conflict-free.
\begin{quote}
  $\FOIS[\phi(x)]:=\forall x \forall y (\phi(x) \land \phi(y)) \rightarrow \lnot Axy$
\end{quote}
Again we note here that the notation $\FOIS[\phi(x)]$ means that the exact definition
of the formula $\FOIS[\phi(x)]$ depends on the formula $\phi(x)$,
e.g., if $\phi(x):=\FOSET[l](x_1,\dots,x_l,x)$ then $\FOIS[\phi(x)]$
is the formula $\forall x \forall y (\FOSET[l](x_1,\dots,x_l,x) \land
\FOSET[l](x_1,\dots,x_l,y)) \rightarrow \lnot Axy$ which in turn
evaluates to $\forall x \forall y (\bigvee_{1 \leq i \leq l}x=x_i \land
\bigvee_{1 \leq i \leq l}y=x_i) \rightarrow \lnot Axy$.

The formula $\FODIFF[\phi_1(x),\phi_2(x)](y)$ is satisfied if and
only if the argument $y$ is contained in the symmetric difference of
the sets of arguments that satisfy the formula $\phi_1(x)$ and the set
of arguments that satisfy the formula $\phi_2(x)$.
\begin{quote}
  $\FODIFF[\phi_1(x),\phi_2(x)](y):=(\phi_1(y) \land \lnot \phi_2(y)) \lor (\lnot \phi_1(y) \land \phi_2(y))$
\end{quote}
The formula $\FOATMOSTK[\phi(x),k]$ is satisfied if and
only if the set of arguments that satisfy the formula $\phi(x)$
contains at most $k$ arguments.
\begin{quote}
  $\FOATMOSTK[\phi(x),k]:=\lnot (\exists x_1,\dots,\exists x_{k+1}
  (\bigwedge_{1 \leq i < j \leq k+1}x_i \neq x_j)  \land
 (\bigwedge_{1
    \leq i \leq k+1}\phi(x_i)))$
\end{quote}

The following formulas represent the semantics \ADM, \COM, \STB. These
formulas are therefore evaluated over a structure with universe $X$
and at least $1$ binary relation $A$ representing an AF
$F:=(X,A)$.

The formula $\FOADM[\phi(x)]$ is satisfied by the structure
representing an AF $F$ if and
only if the set of arguments that satisfy the formula $\phi(x)$ is an
admissible extension of $F$.
\begin{quote}
  $\FOADM[\phi(x)]:=\FOIS[\phi(x)] \land (\forall x \forall z (\phi(x) \land (\lnot \phi(z)) \land Azx)
  \rightarrow (\exists y \phi(y) \land Ayz))$
\end{quote}
The formula $\FOCOM[\phi(x)]$ is satisfied by the structure
representing an AF $F$ if and
only if the set of arguments that satisfy the formula $\phi(x)$ is a
complete extension of $F$.
\begin{quote}
  $\FOCOM[\phi(x)]:=\FOADM[\phi(x)] \land  
 (\forall z ((\forall a Aaz \rightarrow
  \exists x \phi(x) \land Axa) \land (\forall x \phi(x) \rightarrow \lnot (Axz
  \lor Azx))) \rightarrow \phi(z)$
\end{quote}
The formula $\FOSTB[\phi(x)]$ is satisfied by the structure
representing an AF $F$ if and
only if the set of arguments that satisfy the formula $\phi(x)$ is a
stable extension of $F$.
\begin{quote}
  $\FOSTB[\phi(x)]:=\FOIS[\phi(x)] \land (\forall z \phi(z) \lor (\exists a \phi(a) \land Aaz))$
\end{quote}

We are now ready to define the formulas that represent the
problems \Small{}, \Repair{}, \Adjust{}, and \Center{}.

Let $\sigma \in \{\ADM,\COM,\STB\}$. The formula $\FOSMALL[\sigma,k]$
is satisfied by the structure representing an instance $(F,k)$ of \Small{} if and
only if the AF $F$ has 
a non-empty $\sigma$-extension that contains at most $k$ arguments, i.e., if and
only if $(F,k)$ is a \textsc{Yes} instance of \Small{}.
\begin{quote}
  $\FOSMALL[\sigma,k]:=\exists x_1, \dots, \exists x_k \sigma[\FOSET[k](x_1,\dots,x_k,x)]$
\end{quote}
The formula $\FOREPAIR[\sigma,k]$ is satisfied by the structure
representing an instance $(F,S,k)$ of \Repair{} if and
only if $F$ has a $E \in \sigma(F)$ with $\bsdist(E,S) \leq k$, i.e., if and
only if $(F,S,k)$ is a \textsc{Yes} instance of \Repair{}.
\begin{quote}
  $\FOREPAIR[\sigma,k]:=\exists x_1, \dots, \exists x_k
  \sigma[\FODIFF[Sx,\FOSET[k](x_1,\dots,x_k,x)]]$
\end{quote}
The formula $\FOADJUST[\sigma,k]$ is satisfied by the structure
representing an instance $(F,E_0,t,k)$ of \Adjust{} if and
only if $F$ has a $E \in \sigma(F)$ such that $\bsdist(E_0,E) \leq k$ and $t \in E
\bigtriangleup E_0$, i.e., if and
only if $(F,E_0,t,k)$ is a \textsc{Yes} instance of \Adjust{}.
\begin{quote}
  $\FOADJUST[\sigma,k]:=\exists t \exists x_1, \dots, \exists x_{k-1}
  Tt \land \sigma[\FODIFF[E_0x,\FOSET[k](t,x_1,\dots,x_{k-1},x)]]$
\end{quote}
The formula $\FOCENTER[\sigma,k]$ is satisfied by the structure
representing an instance $(F,E_1,E_2)$ of \Center{} if and
only if $F$
has a $E \in \sigma(F)$ with $\bsdist(E_i,E) < \bsdist(E_1,E_2)$ for
every $i \in \{1,2\}$, i.e., if and
only if $(F,E_1,E_2)$ is a \textsc{Yes} instance of \Center{}.

\begin{quote}
  $\FOCENTER[\sigma,k]:=$

\quad $\exists x_1, \dots, \exists x_{k-1}
  \sigma[\FODIFF[E_1x,\FOSET[k-1](x_1,\dots,x_{k-1},x)]] \land$ 
 
\quad $ \FOATMOSTK[k-1,
\FODIFF[\FODIFF[E_1x,\FOSET[k-1](x_1,\dots,x_{k-1},x)],E_2x]$
\end{quote}

Because the length of the above FO formulas is easily seen to be
bounded in terms of the parameter $k$ of the respective problem, these
formulas together with
Proposition~\ref{pro:seese} immediately imply Theorem~\ref{the:tract}.

\section{Concluding Remarks}

We studied the computational problems \textsc{Repair},
\textsc{Adjust}, and \textsc{Center} which arise in the context of
dynamic changes of argumentation systems.  All three problems ask
whether there exists an extension of small distance to some given set
of arguments, and an upper bound to that distance is taken as the
parameter.  We considered all three problems with respect to five
popular semantics: the admissible, the complete, the preferred, the
semi-stable, and the stable semantics, with unrestricted argumentation
frameworks and for argumentation frameworks of bounded degree. We have
determined whether the problems remain coNP-hard, W[1]-hard, or are
fixed-parameter
tractable, 
see 
Figure~\ref{fig:comp-results}.


Parameterized complexity aspects of \emph{incremental computation} have
recently become the subject of research
\cite{DowneyEganFellowsRosamondShaw13,HartungNiedermeier}. We would
like to point out that some of our results, in particular our results
for the \textsc{Repair} problem, can be considered as 
contributions to this line of research: The argumentation framework
has changed, and the existing extension is not anymore an extension
with respect to the semantics under consideration. When considering
the admissible, the complete, and the stable semantics, and when the
degree of the argumentation framework is small, the it is more
efficient to repair the existing extension than to compute an
extension from scratch. On the other hand, when considering the
preferred and the semi-stable semantics, the problems remain
intractable even when the degree is small.

We close by suggesting an ``opportunistic'' version of the \textsc{Repair} problem. That is, given a set of arguments together with an argumentation framework, is it possible to change the framework so that the set becomes an extension? While the allowed elementary changes in the framework can be defined in various ways, the number of such changes needs to be small. Such a problem is a natural candidate for parameterized complexity analysis.

\section*{Acknowledgment}

We would like to thank Stefan Woltran for stimulating discussions.


\end{document}